\documentclass[aps,prb,showpacs,floatfix,nofootinbib,superscriptaddress,twocolumn]{revtex4-1}
\usepackage{amssymb,amsmath,amsfonts,amsthm}
\usepackage[utf8]{inputenc}
\usepackage{times}
\usepackage{graphicx}
\usepackage{verbatim}
\newcommand{\bbone}{\openone}
\newcommand{\bra}[1]{\mbox{$\langle #1|$}}
\newcommand{\ket}[1]{\mbox{$|#1\rangle$}}

\newcommand{\ketbra}[2]{\mbox{$|#1\rangle\langle #2|$}}
\newcommand{\projector}[1]{\mbox{$|#1\rangle\langle #1|$}}

\DeclareMathOperator{\Span}{span}
\newcommand{\LH}{{\sc local Hamiltonian}}
\newcommand{\comm}[2]{[#1, \;#2]}

\newcommand{\QMA}{{\sf QMA}}

\newcommand{\eg}{e.g.~}
\newcommand{\ie}{i.e.~}

\newcommand{\Had}{\mathbb{H}_P}

\newcommand{\tr}{\operatornamewithlimits{Tr}}
\newcommand{\I}{\mathbf{1}}

\newcommand{\Figref}[1]{Fig.~\ref{#1}}

\newcommand{\be}{\begin{equation}}
\newcommand{\ee}{\end{equation}}
\usepackage{hyperref}
\hypersetup{
  pdftitle={Adiabatic Quantum Simulators},
  pdfsubject={Adiabatic Quantum Simulators},
  pdfkeywords={quantum simulator, adiabatic quantum computing, phase
estimation,
Hamiltonian gadget,  Hamiltonian simulation, Aspuru-Guzik, Biamonte,
Whitfield,
Bergholm, Fitzsimons}
}

\newtheorem{lemma}{Lemma}

\begin{document}
\title{Adiabatic Quantum Simulators}
\author{J.D. Biamonte}
\affiliation{Department of Chemistry and Chemical Biology, Harvard University, Cambridge, MA 02138, USA}
\affiliation{Oxford University Computing Laboratory, Oxford OX1 3QD, UK}

\author{V. Bergholm}
\affiliation{Department of Chemistry and Chemical Biology, Harvard University, Cambridge, MA 02138, USA}

\author{J.D. Whitfield}
\affiliation{Department of Chemistry and Chemical Biology, Harvard University, Cambridge, MA 02138, USA}

\author{J. Fitzsimons}
\affiliation{Department of Materials, University of Oxford, Oxford OX1 3PH, UK}
\affiliation{Institute for Quantum Computing, University of Waterloo, Waterloo, Ontario, Canada}

\author{A. Aspuru-Guzik}
\email{aspuru@chemistry.harvard.edu}
\homepage{http://aspuru.chem.harvard.edu}
\affiliation{Department of Chemistry and Chemical Biology, Harvard University, Cambridge, MA 02138, USA}


\begin{abstract}

In his famous 1981 talk, Feynman proposed that unlike classical
computers, which would presumably experience an exponential slowdown when
simulating quantum phenomena, a universal quantum simulator would not.
An ideal quantum simulator would be controllable, and
built using existing technology.
In some cases, moving away from gate-model-based implementations of
quantum computing may offer a more feasible solution for particular experimental implementations.
Here we consider an adiabatic quantum simulator which simulates the
ground state properties of sparse Hamiltonians consisting of one- and
two-local interaction terms, using sparse Hamiltonians with at most
three-local interactions.
Properties of such Hamiltonians can be well approximated with Hamiltonians containing only
two-local terms.
The register holding the simulated ground state is brought
adiabatically into interaction with a probe qubit, followed by a
single diabatic gate operation on the probe which then undergoes free
evolution until measured.
This allows one to recover e.g. the ground state energy of the
Hamiltonian being simulated.
Given a ground state, this scheme can be used to verify the \QMA-complete
problem \LH{}, and is therefore likely more powerful 
than classical computing.
\end{abstract}
\maketitle

\section{Introduction}

Computer simulation of quantum mechanical systems is an indispensable tool
in all physical sciences dealing with nanoscale phenomena.
Except for specific and rare cases~\cite{verstraete08}, classical
computers have not been able to efficiently simulate quantum systems, as in all known techniques at
least one of the 
computational resources required to perform the simulation scales
exponentially with the size of the system being simulated.

Numerous classical approximative methods, such as density functional
theory~(DFT~\cite{ParrYang}) and quantum Monte Carlo
(QMC~\cite{WMH09}) have been developed to address various
aspects of the efficiency problem, but no known polynomially scaling methods are universally
applicable.  Each suffers from particular deficiencies
such as the fermionic sign problem of QMC or the approximate
exchange-correlation functionals of DFT.  
Quantum computers on the other hand, as conjectured by
Feynman~\cite{Feynman1982}, may be used to simulate other quantum
mechanical systems efficiently.
Feynman's conjecture was subsequently 
expanded leading to the rapidly growing area
of study known as
\textit{quantum simulation}~\cite{Lloyd1996,Lidar1999,Somma2002b,
Laflamme2004,ADLH2005,Brown2006,Lanyon2008Sub,whitfield-2010, 2010arXiv1008.4116M}.

Quantum simulation is expected to be able to produce classically
unattainable results in feasible run times, using only a modest number of
fault tolerant (or error corrected) quantum bits.
For example, calculating the ground state energy of the water molecule to
the level of precision necessary for experimental predictions
($\approx 1$ kcal/mol) --- a problem barely solvable on current
supercomputers~\cite{chan_exact_2003} --- would require roughly~128
coherent quantum bits (before error correction)~\cite{ADLH2005}, 
and on the order of billions of quantum gates~\cite{CBMG09}.
To date, several experimental implementations of  
quantum simulation algorithms have been done for small systems~\cite{Lanyon2008Sub,FSG+08,GKZ09}.

Theoretical quantum simulation falls into two main categories:
\textit{dynamic
evolution}
and \textit{static properties}. Both categories
rely
heavily on the Trotter decomposition to handle
non-commuting terms in the Hamiltonian when mimicking
the unitary time propagator of the system to be simulated.
To approximate evolution under a Hamiltonian $H = \sum_{i=1}^k H_i$
consisting of $k$~non-commuting but local terms~$\{H_i\}_{i=1}^k$,
one applies the sequence of unitary gates
$\{e^{-iH_it/n}\}_{i=1}^k$ a total of $n$ times.
As the number of repetitions~$n$ 
tends to infinity the approximation error caused by the non-commutativity
vanishes and the approximation converges to the exact
result~\cite{Lloyd1996}. If each time slice requires a constant number of
gates independent of the
parameter~$t/n$, then reducing the approximation error by repeating the sequence
$n$~times can become expensive for high accuracy applications~\cite{CBMG09}.

Constructing a practical method of quantum simulation is a significant
challenge. Gate-model based simulation methods (with quantum error
correction) can provide a scalable solution but are well out of reach
of present-day experiments, except for small systems.
On the other hand, some experimental implementations seem to be well suited to
operate as adiabatic processors.
For these setups, moving away from the gate model may offer a less
resource-intensive, and consequently a more feasible solution
for simulating medium-sized quantum systems.
This paper addresses the problem by presenting a 
hybrid model of quantum simulation,
consisting of an adiabatically controlled simulation
register coupled to a single gate-model readout qubit.
Our scheme can simulate the constant observables
of arbitrary spin graph Hamiltonians.
It allows the measurement of the expectation value of any
constant $k$-local observable using a $k+1$-local measurement Hamiltonian. 
Fault tolerance of the adiabatic model of quantum computing is a topic of growing interest.  
Its robustness in the presence of noise has been studied in~\cite{CFP01,taqc2008}. 
We only consider the simulation protocol here, and don't study the
adiabatic error correction that would be required for a
practical large scale implementation.

In most quantum computing architectures, the natural interactions are
two-local. However, under certain conditions $k$-local
interactions can be well approximated using techniques such as
perturbative Hamiltonian gadgets~\cite{kempe2006,oliveira2008,biamonte-2008-78, PhysRevA.77.052331}
or average Hamiltonian methods~\cite{waugh1968approach} --- this
provides even the possibility of utilizing gate model fault tolerance
to protect the slower adiabatic evolution, an approach that seems promising.  We note that reference~\cite{PhysRevLett.101.070503} considered the mapping of
a given $n$-qubit target Hamiltonian with $k$-local interactions onto a
simulator Hamiltonian with two-local interactions.

\textbf{Structure of this paper:}
We will continue by giving an overview of the simulator design,
including initialization, adiabatic evolution and
measurement.
The following section investigates the performance of the method in
simulating a small system in the presence of noise.
Finally, we will present our conclusions.
Appendix~\ref{sec:suppmat} further details the specifics of our method
(including the numeric simulation of the technique).

\section{Simulator overview}

We consider the simulation of systems represented by finite collections
of spins acted on by a time-independent Hamiltonian described by a graph
$G = (V, E)$ --- \eg~the Heisenberg and Ising models.
Each graph vertex~$v\in V$ corresponds to a spin acted on by a local
Hamiltonian $L_{v}$, and each edge~$e\in E$ to a
two-local interaction~$K_{e}$ between the involved vertices.
The Hamiltonian~$H_{T}$ we wish to simulate is given by
\begin{equation}
H_{T} = \sum_{v \in V} L_v +\sum_{e \in E} K_e.
\end{equation}

The simulator consists of an adiabatically controlled
simulation register~$S$ with Hamiltonian~$H_S$, and a probe register~$P$
which will be acted on by gate operations and measured projectively. 
We will engineer the probe~$P$ such that it behaves as a controllable
two-level system with the orthonormal basis $\{\ket{p_0}, \ket{p_1}\}$.
Without loss of generality, the probe Hamiltonian can be expressed as
$H_P = \delta \projector{p_1}$, where~$\delta$ is the spectral gap
between the probe's ground and first excited states.

\textit{Initialization:}
We will first set the simulation register Hamiltonian~$H_S$
to~$H_{S,I}$ and prepare $S$ and $P$ in their respective
ground states.  The Hamiltonian~$H_{S,I}$ has a simple
ground state which can be (i) computed classically and (ii) prepared
experimentally in polynomial time --- such as a classical approximation to
the
ground
state of the simulated system. 
The simulator Hamiltonian is thus initially given by
\begin{equation}
H_0 = H_{S,I} \otimes \bbone_P + \bbone_S \otimes H_P.
\end{equation}

\textit{Adiabatic evolution:}
According to the adiabatic theorem~\cite{FGGS00},
a quantum system prepared in an energy eigenstate will remain near 
the corresponding instantaneous eigenstate of the time-dependent
Hamiltonian governing the evolution if there are no level crossings
and the Hamiltonian varies slowly enough.
By adjusting the simulator parameters, we adiabatically change
$H_S$ from $H_{S,I}$ to $H_{S,T}$, the fully interacting Hamiltonian of the
system to be
simulated.

Let us denote the ground state of~$H_{S,T}$ as~$\ket{s_{0}}$. 
At the end of a successful adiabatic evolution $P$~is still in its ground
state~$\ket{p_0}$, 
and $S$ is in (a good approximation to) the ground state of the simulated
system, $\ket{s_0}$. Hence the simulator is now in the ground state
$\ket{\text{g}} = \ket{s_0}\otimes\ket{p_0}$
of its instantaneous Hamiltonian
\begin{equation}
H_1 = H_{S,T} \otimes \bbone_P + \bbone_S \otimes H_P.
\end{equation}

The computational complexity of preparing ground states of quantum systems
has been studied~\cite{kempe2006,oliveira2008,biamonte-2008-78}.  It is
possible to prepare a desired ground state efficiently provided that
the gap between the ground and excited states is sufficiently
large~\cite{FGGS00} (see alternative methods in \cite{PhysRevLett.103.120504}). This depends on the initial and final
Hamiltonians and on the adiabatic path taken.  In general, finding the
ground state energy of a Hamiltonian, even when restricted to certain
simple models, is known to be complete for QMA, the quantum analogue
of the class NP~\cite{kempe2006,oliveira2008,biamonte-2008-78}.
In fact there are physical systems such as spin glasses
in nature which may never settle into their ground states.

However, a host of realistic systems (\eg insulators, molecular systems) can on
physical grounds be expected to retain a large energy gap and should thus be
amenable to quantum simulation algorithms which rely on adiabatic
state preparation.  

\textit{Measurement:}
The measurement procedure begins by bringing $S$~and $P$ adiabatically
into interaction.
The simulator Hamiltonian becomes
\begin{equation}
H_2 = H_1 + \underbrace{A \otimes \ketbra{p_1}{p_1}}_{H_{SP}},
\end{equation}
where the operator~$A$ corresponds to an observable of the simulated system
that is a constant of motion, \ie, commutes with~$H_{S,T}$.
Hence the total energy itself can always be measured by choosing
$A = H_{S,T}$.
Other such observables depend on the particular system and
often can be analytically constructed given the Hamiltonian.

Let us use $a_{\text{min}}$ to denote the lowest eigenvalue of~$A$.
If $a_{\text{min}} +\delta > 0$,
then $\ket{\text{g}}$ is also the ground state of~$H_2$
and in the absence of noise the transitions from the ground state are
perfectly suppressed during the adiabatic evolution
(see Appendix~\ref{sec:suppmat} for proof).
Assuming that $A$ can be decomposed into a sum of
two-local operators, the interaction term $H_{SP}$ involves
three-local interactions. These terms can be implemented using either
Hamiltonian gadget techniques or average Hamiltonians
(see Appendix~\ref{sec:suppmat} for details).

After the adiabatic evolution, at time $t=0$, we apply a Hadamard gate to the measurement probe 
which puts it into a superposition of its two lowest states.
This is no longer an eigenstate of $H_2$, and the system will
evolve as
\begin{equation}
\ket{\psi(t)} = \frac{1}{\sqrt{2}} \ket{s_0} \otimes (\ket{p_0} +
e^{-i  \omega t}\ket{p_1}), 
\end{equation}
where $\omega := (a_0+\delta)/\hbar$,
and $a_0 := \bra{s_0} A \ket{s_0}$
is the expectation value we wish to measure.
We have thus encoded the
quantity~$a_0$, a property of the ground state of~$H_{S,T}$, into the
time dependence of the probe~$P$.

After a time $t$, we again apply a Hadamard gate to the probe,
resulting in the state
\begin{equation}\label{eqn:phase}
\ket{\psi(t)} = \ket{s_0} \otimes \left(\cos \left(\omega t /2\right)
\ket{p_0} +i \sin \left(\omega t /2\right) \ket{p_1}\right),
\end{equation}
and then measure the probe. The probability of finding it in the state
$\ket{p_0}$ is
\begin{equation}
\label{eqn:prob}
P_0(t) = \frac{1}{2}\left(1 + \cos(\omega t)\right) =
\cos^2\left(\omega t /2\right).
\end{equation}
If we have non-demolition measurements (see \eg~\cite{siddiqi2005}) at
our disposal, then measuring the probe does not disturb the state of
the simulator which can
be reused for another measurement. 

One repeats the measurement with different values of~$t$ until
sufficient statistics have been accumulated to reconstruct~$\omega$
and hence~$a_0$ --- this is reminiscent of Ramsey
spectroscopy~\cite{Ram63} and hence should seem natural to
experimentalists. In essence, we have performed Kitaev's phase estimation
algorithm~\cite{Kitaev1995}, using the interaction
Hamiltonian~$H_{SP}$ instead of a controlled unitary.

If the ground state subspace of~$H_{S,T}$ is degenerate and overlaps
more than one eigenspace of~$A$, or
the simulation register~$S$ has excitations to higher
energy states at the beginning of the measurement phase,
the probability of finding the probe in the state $\ket{p_0}$ is given
by a superposition of harmonic modes. For example, for a thermalized state
with inverse temperature~$\beta$, we obtain
\begin{equation}
\label{eqn:prob2}
P_0(t) = \frac{1}{2}\left(1 +\frac{1}{\sum_{xy} e^{-\beta E_x}} \sum_{kl} e^{-\beta E_k} \cos(\omega_{k,l} t)\right),
\end{equation}
where the first summation index runs over the energy eigenstates and
the second over the eigenstates of
$A$ in which energy has value~$E_k$, and $\omega_{k,l} = (a_{k,l}+\delta)/\hbar$.

\section{Effects of noise}
\label{sec:noise}

Any large scale implementation of the proposed method would most likely require
adiabatic error correction~\cite{Jordan2006,Lidar2008,AAN09}.
However, for small systems this might not be necessary indicating that our
method is feasible for immediate experimental technologies. For this reason, we will examine how robust a small scale implementation of our protocol would be.

To assess the effects of noise on the simulation method, we performed
a numerical simulation of the simplest nontrivial implementation of the
hybrid simulator, consisting of two simulator qubits and one probe
qubit, with a randomly chosen~$H_T$.
Each qubit was coupled to its own bosonic heat bath with an Ohmic
spectral density using the Born-Markov approximation~\cite{BP}.
The qubit-bath couplings were chosen such that the resulting single-qubit
decoherence times $T_1$ and $T_2$ are compatible with
recent superconducting flux qubit experiments with fully tunable
couplings (see for example~\cite{niskanen2007}).
The noise model is described further in Appendix~\ref{sec:suppmat}.
The observable~$A$ was chosen to be $H_{S,T}$, the simulated
Hamiltonian itself, which means that the ground state energy was being
measured.

We simulated measurements on 40 evenly distributed values of the time
delay~$t$. At each $t_i$ we performed 50 measurements, and averaged
the results to get an estimate of~$P_0(t_i)$.
An exponentially decaying scaled cosine function was then fitted to
this set of data points to obtain an estimate for~$\omega$ and thus for~$s_0$.
The fit was done using the MATLAB {\sc lsqnonlin} algorithm, guided
only by fixed order-of-magnitude initial guesses for the parameter values.

\begin{figure}[h]
\includegraphics[width=0.5\textwidth]{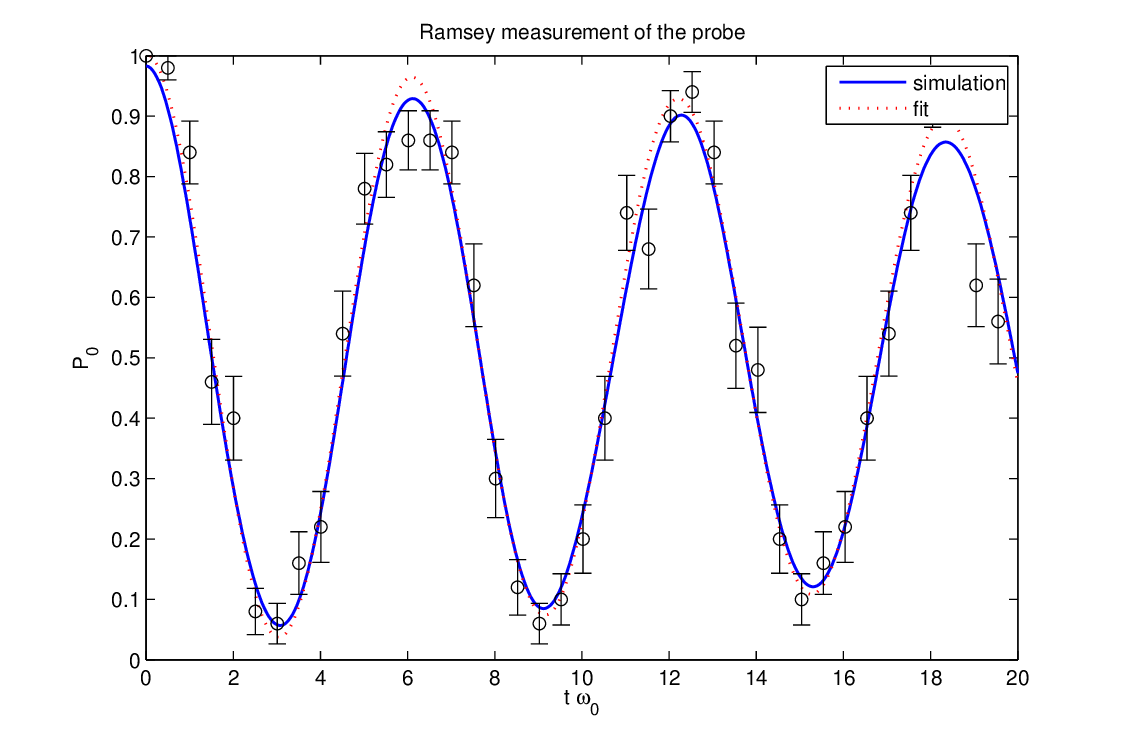}
\caption{Measurement procedure under Markovian noise.
The continuous curve represents the probability of finding
the probe in the state~$\ket{p_0}$, the circles
averaged measurement results and the dotted line a least squares fit to them.
$\hbar \omega_0 = h \cdot 25$~MHz is the energy scale of the Hamiltonians involved.
}
\label{fig:noise}
\end{figure}

The results of the simulation are presented in~\Figref{fig:noise}.
The noise, together with the slight nonadiabaticity of
the evolution, cause excitations out of the ground state
which result in a signal consisting of multiple harmonic
modes. However, the ground state mode still dominates and with a
realistic level of noise and relatively modest
statistics we are able to
reconstruct $\omega$ to a relative precision of
better than~$0.01$. 
This includes both the uncertainty due to finite statistics, and the
errors introduced by the environment through decoherence and the Lamb-Stark
shift.

In an experimental implementation there may also be other noise
sources not considered here related to the measurement process itself,
as well as systematic (hardware) errors in the quantum processor
(qubits, couplers etc.).
Nonetheless, our results indicate that a simple experimental
implementation of the simulation scheme could be possible using
existing hardware.
For future implementations, there exist
fault tolerant constructions for adiabatic quantum
computing~\cite{Jordan2006,Lidar2008,AAN09, PhysRevLett.103.120504, 2010PhRvE..82c1106C}.

\section{Conclusion}\label{sec:conclusion}

The presented simulation scheme differs significantly from existing
gate-model methods.
Instead of a series of coherent gate operations, it uses
an adiabatic control sequence which may
require less complicated control hardware. 
At the measurement stage we require single-qubit gate operations
and measurements, but only on the probe qubit. These operations should be relatively
simple to implement compared to a full Trotter decomposition of the
simulated Hamiltonian.
Without error correction our method has limited scalability, but it
might outperform gate-model simulators in some
small-to-medium-sized problems.
A simple experimental implementation could be
feasible with present-day technology.
In order to simulate a system of $n$~qubits with a two-local Hamiltonian
described by the graph~$G$,
ideally our scheme requires 
one probe qubit for the readout and
$n$~simulation qubits.
Additionally, if Hamiltonian gadgets are used to implement three-local
interactions,
one ancilla qubit is required for each two-local term in the
simulated observable (represented by an edge in the
corresponding graph).
In practice the number of ancillas may be slightly higher if more involved
types of gadgets are used to implement the three-local interactions.
The total number of qubits required thus scales as $O(n)$ for sparse
graphs and $O(n^2)$ for maximally connected ones.

\begin{acknowledgments}
We thank Patrick Rebentrost, Man-Hong Yung, Viv Kendon and Terry Rudolph for helpful discussions.
This work received funding from the Faculty of Arts and Sciences at
Harvard University (JDB and VB), EPSRC grant EP/G003017/1 (JDB), Merton College (JF)
and Academy of Finland (VB).
VB visited Oxford using financial support from the EPSRC grant.  
AA-G and VB thank DARPA under the Young Faculty Award N66001-09-1-2101-DOD35CAP, the Camille and
Henry Dreyfus Foundation, and the Sloan Foundation for support.
AA-G and JDW thank ARO under contract~{W911-NF-07-0304}.
\end{acknowledgments}

\onecolumngrid
\appendix 

\section{Supporting material}
\label{sec:suppmat}
Here we provide supplementary information on
the adiabatic quantum simulation method.
In addition, we have created a MATLAB-based numerical simulation of
a simple instance of the hybrid simulator subject to Markovian noise.
The full source code of the simulation is available on request.

\subsection{Gadgets}
One possible way to effect the 3-local Hamiltonians we require in our
construction is to approximate them with 2-local interactions.

\subsubsection{Perturbative gadget}
This \textit{Gadget Hamiltonian} construction was proposed in~\cite{kempe2006} and has since been used
and extended by others~\cite{oliveira2008,biamonte-2008-78}.
These papers contain further background details including notation.

We label the working qubits~$1$--$3$.  The effective 3-local interaction
is well approximated on qubits $1$, $2$ and $3$ (within~$\epsilon$) in a
low energy subspace --- constructed by adding a penalty Hamiltonian to
an ancillary (mediator) qubit~$m$.  The mediator qubit~$m$ doubles the
size of the state space and the penalty Hamiltonian splits the Hilbert
space into low and high energy subspaces, separated by an energy
gap~$\Delta$ (which is inversely proportional to a polynomial in~$\epsilon$
--- in our case reduced to $\Delta \geq \epsilon^{-3}$).
The details of the gadget we develop for use in this work follow. 

We will apply a Hamiltonian $H_p$ to the mediator qubit~$m$ as well as a
Hamiltonian $V$ to qubits $1$--$3$ and~$m$.  The Hamiltonian~$H_p + V$
has a ground state that is $\epsilon$-close to the desired operator
$JA_1\otimes A_2 \otimes A_3 \otimes \ketbra{0}{0}$.
\be
H_p := \Delta \ket{1}\bra{1}_m 
\ee 
\be 
V(J,\Delta(\varepsilon)) := \textbf{y} + \Delta^{1/3}\ket{0}\bra{0}_m - \Delta^{1/3}A_1\otimes A_2 + \frac{\Delta^{2/3}}{\sqrt{2}}(A_2-A_1)\otimes\sigma_x 
+ JA_3\otimes(\I-\Delta^{2/3}\ket{1}\bra{1}_m)
\ee 
where $\textbf{y}$ is some Hamiltonian already acting on qubits $1-3$ as well as a
possible larger Hilbert space.  Note the gadget above assumes $A_i^2=\I,~\forall i=1,2,3$.  

The so called \textit{self-energy expansion} \eqref{eqn:se} under
appropriate conditions is known to provide a series approximation to the
low-lying eigenspace of an operator. To verify that the Hamiltonian $H_p + V$
gives the desired approximation, one relies on expansion of the self
energy to $4^{th}$ order (here the higher order terms give rise to
effective interactions greater than second order):
\begin{equation}
\label{eqn:se}
\Sigma_-(z) = \bra{0}V\ket{0}
+\frac{\bra{0}V\ket{1}\bra{1}V\ket{0}}{z-\Delta}
+\frac{\bra{0}V\ket{1}\bra{1}V\ket{1}\bra{1}V\ket{0}}{(z-\Delta)^2}
+\mathcal{O}\left(\frac{\|V\|^4}{\Delta^3}\right),
\end{equation}
where the operator is written in the $\{\ket{0},\ket{1}\}$ basis of~$m$.
One considers the range $|z|\leq2|J|+\epsilon$ and notes that
$\|\Sigma_-(z)-JA_1\otimes A_2 \otimes A_3\| = \mathcal{O}(\epsilon)$
and applies the Gadget Theorem~\cite{kempe2006}.


Before concluding this section, we note that care must be taken when adiabatically evolving gadgets.  Reference~\cite{aharonov-2007-37} contains a proof that the linear path Hamiltonian is universal for adiabatic quantum computation.  Universality (and hence a non-exponentially contracting gap) remains when the locality of the construction is reduced using perturbative gadgets~\cite{oliveira2008}.  

\subsubsection{Exact diagonal gadget}
Modern experimental implementations of adiabatic quantum computers
typically are limited to only being able to couple spins with one type
of coupling (\eg $\sigma_z \otimes \sigma_z$).  In such a case, the standard
gadget Hamiltonian approach will not work as these gadgets require
multiple types of couplers~\cite{PhysRevA.77.052331}.  We will now provide
a new type of gadget Hamiltonian which creates an
effective, exact $\sigma_z \otimes \sigma_z \otimes \sigma_z$
interaction in the low energy subspace using just
$\sigma_z \otimes \sigma_z$ and local terms.

Let us first assume that we have access to a penalty function
$H_{\textsc{AND}}(x_*,x_1,x_2)$, where $x_i\in\{0,1\}$ such that
$H_{\textsc{AND}} = 0$ any time $x_* = x_1 x_2$ and is greater than some
large constant $\Delta$ for all $x_* \neq x_1 x_2$.
By solving a system of constraints, such a penalty function is
possible to write as a sum of two-local terms:
\begin{equation}\label{eqn:Hwedge}
H_{\textsc{AND}}(x_*,x_1,x_2) = \Delta(3x_* +x_1 x_2-2x_* x_1-2x_* x_2).
\end{equation}

The Boolean variables $\{x_i\}$ can be represented on qubits using the
correspondence~$x_i \simeq \hat{x}_i = \ketbra{1}{1}_i$. Furthermore, we have
$\sigma^z = \I-2\ketbra{1}{1}$ as usual.
This gives us
\begin{align}
\notag
H &= \sigma^z_1 \sigma^z_2 \sigma^z_3 =
\I -2 \hat{x}_1 -2 \hat{x}_2 -2 \hat{x}_3
+4 \hat{x}_1 \hat{x}_2 +4 \hat{x}_1 \hat{x}_3 +4 \hat{x}_2 \hat{x}_3
-8 \hat{x}_1 \hat{x}_2 \hat{x}_3\\
\label{eq:6}
&\hat{=} \I -2 \hat{x}_1 -2 \hat{x}_2 -2 \hat{x}_3
+4 \hat{x}_1 \hat{x}_2 +4 \hat{x}_1 \hat{x}_3 +4 \hat{x}_2 \hat{x}_3
-8 \hat{x}_* \hat{x}_3 +H_{\textsc{AND}}(x_*,x_1,x_2),
\end{align}
where~\eqref{eq:6} holds in the low energy subspace after the
introduction of an ancilla qubit~$*$.
By writing the Boolean operators in terms of Pauli matrices, we obtain
\begin{equation}
\notag
H = \left(\frac{3}{4} \Delta -1\right) \I
+\left(2-\frac{\Delta}{2}\right) \sigma^z_* +\left(\frac{\Delta}{4}-1\right) \left(\sigma^z_1+\sigma^z_2\right) +\sigma^z_3
-\frac{\Delta}{2} \sigma^z_* \left(\sigma^z_1+\sigma^z_2\right) -2 \sigma^z_* \sigma^z_3
+\left(\frac{\Delta}{4}+1\right) \sigma^z_1 \sigma^z_2 +\left(\sigma^z_1+\sigma^z_2\right) \sigma^z_3.
\end{equation}
We have hence provided a method which allows one to create diagonal
k-local couplings using one- and two-local couplings.  This method opens the
door to simulate a wider range of Hamiltonians using current and next
generation quantum computing technology.  In addition, it should also
prove useful in future investigations into the fault tolerance of the
proposed protocol, in a way similar to the comparison
in~\cite{Brown2006}.

\subsection{Average Hamiltonian Method}
An alternate method of generating the special Hamiltonians we require
is to make use of the time average of a series of simpler generating
Hamiltonians~\cite{waugh1968approach}. It has long been known that by 
regularly switching between a set of fixed Hamiltonians $\{H_i\}$,
it is possible to approximate time evolution under any other 
Hamiltonian $H$, provided that $H$ is contained within the algebra 
generated by $\{H_i\}$. This fact lies at the heart of both average 
Hamiltonian theory and geometric control theory. Over the years many 
methods have been developed for making the approximations accurate 
to high order, however here we will focus only on approximations 
correct to first order.

In order to construct an average Hamiltonian we will make use of a first order approximation to the Baker-Campbell-Hausdorff formula:
\begin{equation}
\log(e^A e^B) \approx A + B + \frac{1}{2}[A,B].
\end{equation}
From this we obtain formulae for approximating the exponential of both the sum and the Lie bracket of $A$ and $B$ to first order.
\begin{eqnarray}
e^{A+B} &\approx& e^{A} e^{B},\\
e^{[A,B]} &\approx& e^{A} e^{B} e^{-A} e^{-B}.
\end{eqnarray}
These equations can be related to time evolution under some
Hamiltonians, $H_i$, by replacing A and B with operators of the form
$i H_i t/\hbar$. Clearly by applying these rules recursively,
it is possible to generate any Hamiltonian in the Lie algebra
generated by the initial set of static Hamiltonians. We note that the
combination of any pairwise entangling Hamiltonian, such as an Ising
Hamiltonian together with tunable local Z and X fields is sufficient
to generate the full Lie algebra $\mathfrak{su}(2^N)$ for an N qubit system, and
so can be used to approximate an arbitrary Hamiltonian.

Although the time-varying Hamiltonian means that the system does not have a ground state in the normal sense, the average energy of the system is minimized when the system is within $O(t)$ of the ground state of the target average Hamiltonian. As a result of this, if the time scale for switching between Hamiltonians is small compared to the time scale for the adiabatic evolution of the system, the system will behave as if it was experiencing the average Hamiltonian.

\subsection{Noise model}

The noise model used in our MATLAB simulation consists of
a separate bosonic heat bath coupled to each of the qubits.
The baths have Ohmic spectral densities,
\be
J(\omega) = \hbar^2 \omega \Theta(\omega) \Theta(\omega_c-\omega),
\ee
where the cutoff~$\omega_c$ was chosen to be above every transition
frequency in
the system, and are assumed to be uncorrelated.
Each bath is coupled to its qubit through an interaction operator of the
form
\be
D = \lambda (\cos(\alpha)\sigma_z +\sin(\alpha)\sigma_x).
\ee
Using the Born-Markov approximation we obtain an
evolution equation which is of the Lindblad form~\cite{BP}.
Denoting the level splitting of a qubit by~$\hbar \Delta$,
we obtain the following uncoupled single-qubit decoherence times:
\begin{align}
T_{1}^{-1} &= \lambda^2 \sin^2(\alpha)
2 \pi \Delta \coth(\beta \hbar \Delta/2),\\
T_{2}^{-1} &= \frac{1}{2} T_{1}^{-1} +\lambda^2 \cos^2(\alpha) 4 \pi / (\hbar
\beta),
\end{align}
where $\beta = \frac{1}{k_B T}$.
Given $T_1$, $T_2$, $T$ and $\Delta$, we can solve the bath coupling
parameters $\lambda$ and $\alpha$ separately for each qubit, and then use
the same noise model in the fully interacting case.
\begin{table}
\caption{Noise model parameters}
\centering
\begin{tabular}{lr}
\hline \hline
$T_1$ & $\sim N(1.0, 0.1)~\mu \text{s}$\\
$T_\phi$ & $\sim N(1.3, 0.1)~\mu \text{s}$\\
$T_2^{-1}$ & $= \frac{1}{2}T_1^{-1} +T_\phi^{-1}$\\
$T$ & 20 mK\\
$\omega_c$ & 20 $\omega_0$\\
$\omega_0$ & $2\pi \cdot 25$ MHz\\
\hline
\end{tabular}
\label{table:par}
\end{table}
The values used for the parameters are presented in Table~\ref{table:par}.
The single-qubit decoherence times $T_1$ and $T_2$ were chosen to be
compatible with recent coupled-qubit experiments with fully tunable
couplings such as~\cite{niskanen2007}.
The bath temperature~$T$ and the energy scale~$\hbar \omega_0$ of
$H_T$, the Hamiltonian to be simulated,
were likewise chosen to match the temperatures and coupling strengths
found in contemporary superconducting qubit experiments.
To simulate manufacturing uncertainties, the actual values of the
$T_1$ and $T_\phi$
parameters for the individual qubits were drawn from a Gaussian
distribution with a small standard deviation.

\subsection{Measurement}

The pre-measurement system Hamiltonian is 
\begin{equation}
\label{eq:H1}
H_{1} = H_{S,T} \otimes \bbone + \bbone \otimes H_P.
\end{equation}
The operators $H_S$ and $H_P$ can be expanded in terms of their
eigenvalues and eigenstates using the (possibly degenerate) spectral
decomposition
\be
H_{S} = \sum_{kj} s_k \projector{s_{k,j}}, \quad s_0 < s_1 < \ldots
\ee
and correspondingly for~$H_P$.

Let the states of systems $S$ and $P$ begin in their ground state
subspaces, 
the full normalized state $\ket{\text{g}}$ of
the simulator belonging to the ground state subspace
of the non-interacting Hamiltonian $H_1$:
\be
\ket{\text{g}} \in \Span \{ \ket{s_{0,k}}\otimes\ket{p_{0,l}} \}_{k,l}.
\ee
$S$~and $P$ are brought adiabatically into interaction with each other.
The Hamiltonian becomes
\begin{equation}
\label{eq:H2}
H_2 = H_1 + \underbrace{A \otimes (\bbone_P -\Pi_{p_0})}_{H_{SP}},
\end{equation}
where the operator~$A$ corresponds to an observable of the simulated system
that commutes with~$H_S$, and
\be
\Pi_{p_0} = \sum_m \projector{p_{0,m}}
\ee
is the projector to
the ground state subspace of $H_P$.
Because $A$ and $H_S$ commute, they have shared eigenstates:
\begin{align}
H_S \ket{s_{k,j}} &= s_k \ket{s_{k,j}},\\
A\ket{s_{k,j}} &= a_{k,j} \ket{s_{k,j}}.
\end{align}

\subsubsection{Ground state lemma}

We will now show that Hamiltonians $H_1$ and $H_2$ have the same
ground state subspace given that~$a_{\text{min}} +p_1-p_0  > 0$
where~$a_{\text{min}}$ is $A$'s lowest eigenvalue.
\begin{lemma}
Let $H_1$ and $H_2$ be the finite dimensional Hamiltonians defined
previously in~\eqref{eq:H1} and~\eqref{eq:H2},
and $a_{\text{min}} +p_1-p_0 > 0$.
Now, iff $H_1 \ket{\star} = \lambda \ket{\star}$, where
$\lambda$ is the smallest eigenvalue of $H_1$,
then
$H_2 \ket{\star} = \kappa \ket{\star}$, where
$\kappa$ is the smallest eigenvalue of $H_2$.
\end{lemma} 

\begin{proof}
Firstly, we have
\begin{align*}
\bra{g} H_1 \ket{g} &= s_0 +p_0,\\
\bra{g} H_{SP} \ket{g} &= 0. 
\end{align*}
Now, expanding an arbitrary normalized state $\ket{\phi}$ in the
eigenbases of $H_S$ and~$H_P$,
\[
\ket{\phi} = \sum_{xyij} c_{xyij} \ket{s_{x,i}}\ket{p_{y,j}},
\]
we get
\begin{align*}
\bra{\phi} H_1 \ket{\phi} &= \sum_{xyij} |c_{xyij}|^2 (s_x +p_y)\\
 &\ge \sum_{xyij} |c_{xyij}|^2 (s_0 +p_0) = s_0 +p_0,\\
\bra{\phi} H_{SP} \ket{\phi} &=  \bra{\phi} \sum_{xyij} c_{xyij}  A
\ket{s_{x,i}} (\bbone_P -\Pi_{p_0}) \ket{p_{y,j}}\\
&= \bra{\phi} \sum_{xij} \sum_{y\ge1} c_{xyij}  a_{x,i} \ket{s_{x,i}}
\ket{p_{y,j}}\\
&= \sum_{xij} \sum_{y\ge1} |c_{xyij}|^2  a_{x,i}\\
&\ge a_{\text{min}} \sum_{xij} \sum_{y\ge1} |c_{xyij}|^2 \ge 0 \quad
\text{if}
\quad a_{\text{min}} \ge 0.
\end{align*}
Hence, if $A$'s lowest eigenvalue~$a_{\text{min}} \ge 0$ then $\ket{g}$ is
the
ground state of~$H_2 = H_1+H_{SP}$ as well. If $A$ is not nonnegative, we
can
perform the transformation
\begin{align*}
H_S' &= H_S +a_{\text{min}} \I,\\
H_P' &= H_P -a_{\text{min}} \Pi_{p_0},\\
A' &= A -a_{\text{min}} \I.
\end{align*}
This leaves $H_2$ invariant, but makes $A'$ nonnegative.
As long as $a_{\text{min}} +p_1-p_0 > 0$ then $\{\ket{p_{0,m}}\}_m$
still span the ground state subspace of $H_P'$ and the above analysis
remains valid.
\end{proof}

\subsubsection{Ground state degeneracy and excitations}

If the ground state subspace of~$H_{S,T}$ is degenerate and overlaps
more than one eigenspace of~$A$, or
the simulation register~$S$ has excitations to higher
energy states at the beginning of the measurement phase, we need a
more involved analysis of the measurement procedure.
Assume the pre-measurement-phase state of the simulator is given by
\begin{align}
\rho_0
&= \left(\sum_{klmn} c_{klmn}\ketbra{s_{k,l}}{s_{m,n}} \right) \otimes 
\begin{pmatrix}
a & 0\\
0 & 1-a
\end{pmatrix}.
\end{align}
All three terms in $H_2$ commute given that $\comm{H_{S}}{A} = 0$.
Hence
\begin{align}
\notag
e^{qH_2}
&= (e^{qH_{S}}\otimes \bbone_P)(\bbone_S \otimes e^{qH_P})(e^{qH_{SP}})\\
&= (e^{qH_{S}} \otimes e^{qH_P})
(e^{qA} \otimes (\bbone_P -\Pi_{p_0}) +\bbone_S \otimes \Pi_{p_0}).
\end{align}

As a result of the measurement procedure right before the actual
measurement the state is given by
\begin{align}
\notag
\rho_1(t)
&= \Had e^{-i t H_2/\hbar} \Had \rho_0 \Had e^{i t H_2/\hbar} \Had\\
\notag
&=
\Had e^{-i t H_{SP}/\hbar} \left(\sum_{klmn} c_{klmn}
e^{-it(s_k-s_m)/\hbar}\ketbra{s_{k,l}}{s_{m,n}} \right) \otimes
\frac{1}{2}
\begin{pmatrix}
1 & e^{it\delta/\hbar}(2a-1) \\
e^{-it\delta/\hbar}(2a-1) & 1
\end{pmatrix}
e^{i t H_{SP}/\hbar} \Had\\
&=
\sum_{klmn} \left(c_{klmn} e^{-it(s_k-s_m)/\hbar}\ketbra{s_{k,l}}{s_{m,n}} \right)
\otimes
\frac{1}{2} \Had
\begin{pmatrix}
1 & e^{it\omega_{m,n}}(2a-1) \\
e^{-it\omega_{k,l}} (2a-1) & e^{-it(a_{k,l}-a_{m,n})/\hbar}
\end{pmatrix}
\Had,
\end{align}
where $\omega_{k,l} = (a_{k,l}+\delta)/\hbar$ and $\Had$ is the
Hadamard gate operating on the probe.
Projecting the probe to the ground state subspace, we get $\Pi_{p_{0}}
\rho_1(t)
\Pi_{p_{0}}=$
\begin{equation}
=\sum_{klmn} \left(c_{klmn} e^{-it(s_k-s_m)/\hbar}\ketbra{s_{k,l}}{s_{m,n}} \right)
\otimes
\frac{1}{4} 
(1 +e^{-it(a_{k,l}-a_{m,n})/\hbar}
+ 2 \cos (\omega_{m,n} t) (2a-1))
\ketbra{0}{0}
\end{equation}
Thus the probability of finding the probe in the ground state subspace is
given
by a superposition of harmonic modes:
\begin{align}
P_{0}(t) &= \tr(\Pi_{p_{0}} \rho_1(t) \Pi_{p_{0}}) =
\frac{1}{2} (1 +\sum_{kl} c_{klkl} \cos(\omega_{k,l} t)(2a-1)).
\end{align}

\bibliography{hybrid}


\end{document}